%% file: bienkowski.tex
\documentclass[proceedings]{stacs}
\stacsheading{2010}{83-94}{Nancy, France}
\firstpageno{83}

\usepackage{amsmath, amsthm, amssymb}
\usepackage{amsfonts}
\usepackage{enumitem}

\newcommand{\ALG}{\textsc{Alg}}

\newcommand{\cA}{\mathcal{A}}
\newcommand{\cE}{\mathcal{E}}
\newcommand{\cP}{\mathcal{P}}
\newcommand{\cR}{\mathcal{R}}
\newcommand{\remove}[1]{}
\newcommand{\emdash}{\hspace{1mm}---\hspace{1mm}}
\newcommand{\ep}{\varepsilon}

\title{Dynamic sharing of a multiple access channel}

\author[lab1]{M.~Bienkowski}{Marcin Bienkowski}
\author[lab2]{M.~Klonowski}{Marek Klonowski}
\author[lab2]{M.~Korzeniowski}{Miroslaw Korzeniowski}
\author[lab3]{D.~R.~Kowalski}{Dariusz R.~Kowalski}

\address[lab1]{Institute of Computer Science, University of Wroc{\l}aw, Poland}
\address[lab2]{Institute of Mathematics and Computer Science, Wroc{\l}aw University of Technology, Poland}
\address[lab3]{Department of Computer Science, University of Liverpool, UK}

\subjclass{C.1.4 Parallel Architectures, C.2.1 Network Architecture and Design, F.2.2 Nonnumerical Algorithms and Problems. Supported by Polish Ministry of Science and Higher
Education grants no N N206 2573 35 and N N206 1723 33, and by the Engineering and Physical
Sciences Research Council [grant number EP/G023018/1].
}

\begin{document}

\maketitle

\input{abstract}

\input{introduction}

\input{lower-bounds}

\input{random}

\input{no-lockout}

%\newpage
%\begin{appendix}
%\input app-lower
%\input app-random
%\end{appendix}

\end{document}

%% file: abstract.tex
In this paper we consider the mutual exclusion problem on a multiple access channel.
Mutual exclusion is one of the fundamental problems in distributed computing.
In the classic version of this problem, $n$ processes perform a concurrent
program which occasionally triggers some of them to use shared resources,
such as memory, communication channel, device, etc.
The goal is to design a distributed algorithm to control entries and exits to/from
the shared resource
%also called a critical section, 
in such a way that in any time
there is at most one process accessing it.
We consider both the classic and a slightly weaker version of mutual exclusion, 
called $\ep$-mutual-exclusion, 
where for each period of a process staying in the critical section the probability that there is some other process in the critical section is at most $\ep$.
%One of the desired properties of mutual exclusion algorithms is fairness,
%i.e., each process that wants to enter the critical section will succeed eventually. 
%We study this property by defining and analyzing competitive latency of 
%mutual exclusion algorithms.
We show that there are channel settings, where the classic mutual exclusion
is not feasible even for randomized algorithms, while $\ep$-mutual-exclusion is. 
In more relaxed channel settings, we prove an exponential gap between 
%competitive latency
the makespan complexity of the classic mutual exclusion problem 
and its weaker $\ep$-exclusion version.
We also show how to guarantee fairness of mutual exclusion algorithms,
%One of the desired properties of mutual exclusion algorithms is fairness,
i.e., that each process that wants to enter the critical section will eventually succeed.

%% file: introduction.tex
\section{Introduction}

In this paper we consider randomized algorithms for mutual exclusion: one of
the fundamental problems in distributed computing. We assume that there are
$n$ different processes labeled from $0$ to $n-1$ communicating  through a
multiple access channel (MAC).  The computation and communication proceed in
synchronous slots, also called {\em rounds}.  In~the mutual exclusion problem,
each process performs a concurrent program and occasionally requires exclusive
access to shared resources. The part of the code corresponding to this
exclusive access is called a {\em critical section}.  The goal is to provide 
a~mechanism that controls entering and exiting the critical section and
guarantees exclusive access at any time. The main challenge is that the
designed mechanism must be universal, in the sense that exclusive access must
be guaranteed regardless of the times of access requests made by other
processes.

\paragraph{\bf Multiple Access Channel (MAC)}

We consider a multiple access channel as both communication medium and the
shared-access device.  As a communication medium, MAC allows each process
either to transmit or listen to the channel at a round,\footnote{Most of the previous work on MAC,
motivated by Ethernet applications, assumed that a process can
transmit and listen simultaneously; our work instead follows the recent trends
of wireless applications where such simultaneous activities are excluded due to
physical constraints.} and moreover, if more than one process transmits, then
a {\em collision} (signal interference) takes place.
Depending on the devices used in the system, there are several additional
settings of MAC that need to be considered.  One of them is the ability of a
process to distinguish between background noise when no process transmits (also
called {\em silence}) and collision.  If such capability is present at each process,
we call the model {\em with collision detection} (CD for short); if no process
has such ability, then we call the setting {\em without collision detection}
(no-CD).  Another feature of the model is a constant access to the global clock (GC for
short) by all processes or no such access by any of them (no-GC).  The third
parameter to be considered is a knowledge of the total number of available
processes $n$ (KN for short) or the lack of it (no-KN).

%Investigating different configurations of the model setting is motivated by different classes of used devices. 
%In order to distinguish between background noise occurring in MAC when no process transmits and a collision that increases the level of noise, more sensitive devices and more elaborated technical solutions are required. Similarly, the knowledge of well-synchronized global time may be assumed only in case of more powerful wireless devices.

%%%%%%%%%%%%%%%%%%%%%%%%%
%%%%%%%%%%%%%%%%%%%%%%%%%

\paragraph{\bf Mutual Exclusion Problem.}

%A mutual exclusion is an access-control kind of problem,
%where a concurrent program run on independent devices
%occasionally requires an access to the shared object. 
%More precisely, 
In this problem, each concurrent process executes a protocol
partitioned into the following four sections:
\begin{description}
\item[\it \ Entry] the part of the protocol executed in preparation for entering the critical section;
\item[\it \ Critical] the part of the protocol to be protected from concurrent execution;
\item[\it \ Exit] the part of the protocol executed on leaving the critical section;
\item[\it \ Remainder] the rest of the protocol.
\end{description}
These sections are executed cyclically in the order: {\em remainder, entry,
critical, and exit}.  Intuitively, the remainder section corresponds to local
computation of a process, and the critical section corresponds to the access to the
shared object (the channel in our case); though the particular purpose and
operations done within each of these sections are not a~part of the problem.
%The goal of the mutual exclusion problem is to develop 
%entry and exit sections, so that for any remainder and critical sections
%several properties of execution are guaranteed
%(they will be specified later).
Sections entry and exit are the parts that control 
switching between remainder and critical sections in a process,
in order to assure some desired properties of the whole system.

In the traditional mutual exclusion problem, as defined in~\cite{AW,Ly} in the
context of shared-memory model, the adversary controls the sections remainder
and critical.  In particular, she controls their duration in each cycle,
subject only to the obvious assumptions that this duration in each cycle is
finite or the last performed section is the remainder one.  The mutual
exclusion algorithm, on the other hand, provides a protocol for the entry and
exit sections of each process.  In this sense, the mutual exclusion problem can
be seen as a game between the adversary controlling the lengths of remainder
and critical sections of each process (each such section for each process may
have different length) and the algorithm controlling entry and exit sections.
The goal of the algorithm is to guarantee several useful properties of the
execution (to be defined later), while the goal of the adversary is to prevent
it.  Note that the sections controlled by the adversary and those controlled by
the algorithm are interleaved in the execution.  Additionally, in order to make
the game fair, it is typically assumed that every variable used by the
algorithm, i.e., in the entry and exit sections, cannot be modified 
by the adversary in the critical and remainder sections, and vice versa, i.e., 
no variables used by the adversary in the remainder and critical sections can be accessed by the algorithm.

In the model of communication over MAC, 
a process in the entry or the exit section can do the following in a single round: 
perform some action on the channel (either transmit a message or listen), 
do some local computation, and change its section
either from entry to critical or from exit to remainder.
We assume that changing sections occurs momentarily
between consecutive rounds, i.e., in each round
a process is exactly in one section of the protocol.

Since a multiple-access channel is both the only communication medium and
the exclusively shared object, additional constraints, different from
the classic ones regarding e.g., shared memory objects, must be imposed:
\begin{itemize}
\item
no process in the remainder section is allowed to transmit on the channel;
\item
a process in the critical section has to transmit a message on the channel
each round until it moves to the exit section, 
and each such message must be labelled {\em critical}; we call them
{\em critical messages}.
\end{itemize}
If any of these conditions was violated, the adversary would have
an unlimited power of creating collisions on the channel, 
and thus preventing any communication. 

A classic mutual exclusion algorithm should satisfy the following three properties for any round
$i$ of its execution:
\begin{description}
\item[\it \ Exclusion] 
at most one process is in the {\em critical} section in round $i$.
\item[\it \ Unobstructed exit]
if a process $p$ is in the exit section at round $i$, then process $p$ 
will switch to the remainder section eventually after~round~$i$.
\item[\it \ No deadlock] 
if there is a process in the entry section at round $i$, then 
{\em some} process will enter the critical section eventually after round $i$.
\end{description}
To strengthen the quality of service guaranteed by mutual exclusion algorithms,
the following property{\emdash}stronger than no-deadlock{\emdash}has been considered: 
\begin{description}
\item[\it \ No lockout]
if a process $p$ is in the entry section at round $i$, then {\em process $p$}
will enter the critical section eventually after~round~$i$.
\end{description}
Note that{\emdash}to some extent{\emdash}this property ensures fairness: 
each process demanding an~access to the critical section will eventually get it.

As we show, in some cases the exclusion condition is impossible or very costly to achieve. 
Therefore, we also consider a slightly weaker condition: 

\begin{description}
\item[\it \ $\ep$-exclusion] 
for every process $p$ and for every time interval in which $p$ is continuously in 
the critical section, 
the probability that in any round of this time interval 
there is another process being in the critical section is at most $\ep$. 
\end{description}
Intuitively, $\ep$-exclusion guarantees mutual exclusion ``locally'', i.e., for every
single execution of the critical section by a process, with probability at least $1-\ep$.
The version of the problem satisfying $\ep$-exclusion condition
is called {\em $\ep$-mutual-exclusion}.

\paragraph{\bf Complexity Measure.}

We use the {\em makespan} measure, as defined in~\cite{CGKP}
in the context of deterministic algorithms.
Makespan of an execution of a given deterministic mutual exclusion algorithm is defined as 
the maximum length of a~time interval in which there is some process in the entry section and 
there is no process in the critical section.
Taking maximum of such values over all possible executions
defines the makespan of the algorithm.
%In this work, we analyze the expected makespan 
%only in the context of simplified versions of our protocols
%that satisfy only no-deadlock property (but no no-lockout property).
In order to define expected makespan, suitable for randomized algorithms
considered in this work, we need more formal
definitions of an adversarial strategy. 
%and latency measure under such strategy.
Let $\cP$ be a strategy of the adversary, defined as
a set of $n$ sequences, where each sequence corresponds to a different process
and contains, subsequently interleaved, lengths of remainder and critical sections
of the corresponding process.
We assume that each sequence is either infinite or of even length;
the latter condition means that after the last critical section 
%(lengths of critical sections are stored on even positions) 
the corresponding process runs the remainder section forever.
For a given mutual exclusion algorithm $\ALG$ and adversarial strategy $\cP$,
we define $L(\ALG,\cP)$ as a random variable equal to the maximum length 
of a time interval in which there is some process in the entry section and 
there is no process in the critical section in an execution
of $\ALG$ run against fixed strategy $\cP$.
The expected makespan of algorithm $\ALG$ is defined as 
the maximum of expected values of $L(\ALG,\cP)$, taken 
over all adversarial strategies $\cP$.
Note that every algorithm with makespan bounded for all executions
satisfies no-deadlock property, but not necessarily no-lockout.
%Note that the competitive latency measure is different from the makespan measure 
%in two ways. First, it is competitive, while the other one is worst-case. 
%Second, it corresponds to the no-lockout property and fairness,
%while the makespan ``quantifies'' only the weaker no-deadlock property.

For the $\ep$-mutual-exclusion problem, defining makespan is a bit more subtle.
We call an execution {\em admissible} if the mutual exclusion property is
always fulfilled, i.e., no two processes are in the critical section in the same round.
Then in the computation of the (expected) makespan, we neglect non-admissible
executions.

\subsection{Our Results}

We consider the mutual exclusion problem and its weaker $\ep$-exclusion version
in the multiple access channel.  Unlike the previous paper~\cite{CGKP}, where
only no-deadlock property was guaranteed, we also focus on fairness. Also in
contrast to the previous work on mutual exclusion on MAC, we mostly study
randomized solutions.  In the case of the mutual exclusion problem, we allow
randomized algorithms to have variable execution time but they have to be
always correct. On the other hand, a randomized solution for the
$\ep$-mutual-exclusion problem is allowed to err with some small probability
$\ep$. Thus, for the former problem, we require Las Vegas type of solution,
whereas for the latter we admit Monte Carlo algorithms.  Note that very small
(e.g., comparable with probability of hardware failure) risk of failure (i.e.,
situation wherein two or more processes are in the critical section at the same
round) is negligible from a~practical point of view. 

We show that for the most severe channel setting, i.e., no-CD, no-GC and no-KN,
mutual exclusion is not feasible even for randomized algorithms (cf.~Section~\ref{sec:lower}).

In a more relaxed setting, there is an exponential gap between the complexity
of the mutual exclusion problem and the $\ep$-mutual-exclusion problem.
Concretely, we prove that the expected makespan of (randomized) solutions for
the mutual exclusion problem in the no-CD setting is $\Omega(n)$, even if the
algorithm knows $n$, has access to the global clock
(cf.~Section~\ref{sec:lower}), and even if only no-deadlock property is
required.  On the other hand, for the $\ep$-mutual-exclusion problem, we
construct a randomized algorithm, requiring only the knowledge of~$n$,
which guarantees no-lockout property, and whose makespan is $O(\log n \cdot \log
(1/\ep))$ (cf.~Sections~\ref{sec:random:kn} and~\ref{sec:lockout}).

When collision detection is available and only no-deadlock property is
required, we show that the makespan of any mutual exclusion algorithm is at
least $\Omega(\log n)$ (cf.~Section~\ref{sec:lower}) and we construct an
algorithm for the $\ep$-mutual-exclusion problem with expected makespan $O(\log
\log n + \log (1/\ep))$ (cf.~Section~\ref{sec:random:collision}).  Further, we
show how to modify this algorithm to guarantee no-lockout property as well; its
expected makespan becomes then $O(\log n + \log (1/\ep))$
(cf.~Sections~\ref{sec:random:collision} and~\ref{sec:lockout}). 

Finally, if we do not require no-lockout property,  we show how to solve the
$\ep$-mutual-exclusion problem in makespan $O(\log n \cdot \log (1/\ep))$,
where only the global clock is available (cf.~Section~\ref{sec:random:clock}).

We also present a generic method that, taking a mutual exclusion algorithm with
no-deadlock property, turns it into the one satisfying stronger no-lockout
condition. This method applied to the deterministic algorithms
from~\cite{CGKP} produces efficient deterministic solutions satisfying the
no-lockout property. 
%Moreover, the complexity of resulting algorithms 
%is given in terms of competitive latency (cf.~Section~\ref{sec:lockout}).

Due to space limitations, the missing details and proofs will appear in the full version of
the paper.

\subsection{Previous and Related Work}

The multiple access channel is a well-studied model of communication. 
In many problems considered in this setting, one of the most important issues
is to assure that successful transmissions occur in the computation.
These problems are often called {\em selection problems}.
They differ from the mutual exclusion problem by the fact that 
they focus on successful transmissions within a bounded length period,
while mutual exclusion provides control mechanism for 
dynamic and possibly unbounded computation.
In particular, it includes recovering from 
long periods of cumulative requests for the critical section as well as from 
long periods containing no request.
%both long periods of cumulative
%requests for the critical section as well as for the long period of no requests.
Additionally, selection problems were considered typically in the context of Ethernet
or combinatorial group testing, and as such they allowed a~process
to transmit and to listen simultaneously, which is not the case in our model
motivated by wireless applications.
Selection problems can be further split into two categories.
In the static selection problems, it is assumed that a subset of processes
become active at the same time and a subset of them must eventually transmit successfully.
Several scenarios and model settings, including parameters considered
in this work such as CD/no-CD, GC/no-GC, KN/no-KN, randomization/determinism, 
were considered in this context, see e.g.,~\cite{BGI,Cap,CMS,GW,JKZ,Kow,KM,NO,TM,Wil}.
In the {\em wake-up} problem, processes are awaken in (possibly) different
rounds and the goal is to assure that there will be a round with successful transmission
(``awakening'' the whole channel) shortly after the first process is awaken, 
see, e.g.,~\cite{CGKR-05,GPP,JS}. 

More  dynamic kinds of problems, such as transmission of dynamically arriving packets, 
were also considered in the context of MAC. In the (dynamic) packet transmission problem, 
the aim is to obtain bounded throughput and bounded latency.
Two models of packet arrival were considered: stochastic (cf.,~\cite{Gold-04}) and
adversarial queuing (cf.,~\cite{Ben-05,CKR}).
There are two substantial differences between these settings and our work.
First, the adversaries imposing dynamic packet arrival are different than the adversary
simulating execution of concurrent protocol. 
Second, as already mentioned in the context of selection problems, these papers 
were inspired by
Ethernet applications where it is typically allowed to transmit and listen simultaneously.

In a very recent paper~\cite{CGKP}  {\em deterministic} algorithms for mutual
exclusion problem in MAC under different settings (CD, GC, KN) were studied.
%with MAC in different settings (CD, GC, KN) were presented.
The authors proved that with none of those three characteristics mutual
exclusion is infeasible.  Moreover, they presented an optimal{\emdash}in terms
of the makespan measure{\emdash}$O(\log n)$ round algorithm for the model with
CD.  They also developed algorithms achieving makespan $O(n\log^2 n)$ in the
models with GC or KN only, which, in view of the lower bound $\Omega(n)$ 
on deterministic solutions proved for any model with no-CD, is close to optimal.  
Our paper differs from~\cite{CGKP} in three ways.
First, we consider both deterministic and randomized solutions.
Second, for the sake of efficiency we introduce the $\ep$-mutual-exclusion problem.
Third, we study fairness of protocols, which means that we consider
also no-lockout property.

%% file: lower-bounds.tex
\section{Lower Bounds for the Mutual Exclusion Problem}
\label{sec:lower}

In our lower bounds, we use the concept of transmission schedules to capture
transmission/listening activity of processes in the entry or exit section.
Transmission schedule of a~process $p$ can be regarded as a binary sequence
$\pi_p$ describing  the subsequent communication actions of the process. The
sequence can be finite or infinite.  For  non-negative integer $i$,
$\pi_p(i)=1$ means that process $p$ transmits in round~$i$ after starting its
current section, while  $\pi_p(i)=0$  means that the process listens in round
$i$.  We assume that round~$0$ is the round in which the process starts its
current run of the entry or the exit section.

The following results extend the lower bounds and impossibility 
results for deterministic mutual exclusion proved in \cite{CGKP} to  
randomized solutions. All the presented lower bounds work 
even if we do not require no-lockout, but a weaker no-deadlock property.

\begin{theorem}
\label{t:lower-not-feasible}
There is no randomized mutual exclusion algorithm with 
no-deadlock property holding with a positive probability in the setting 
without collision detection, without global clock and without 
knowledge of the number $n$ of processes.
\end{theorem}

\begin{theorem}
\label{t:lower-log}
The expected makespan of any randomized mutual exclusion algorithm 
is at least $\log n$, even in the setting 
with collision detection, with global clock and with 
knowledge of the number $n$ of processes.
\end{theorem}

\begin{theorem}
\label{t:lower-n}
The expected makespan of any randomized mutual exclusion algorithm 
is at least $n/2$ in the absence of collision detection capability, even in the setting 
with global clock and with knowledge of the number $n$ of processes.
\end{theorem}

\begin{proof}
To arrive at a contradiction, let $\cR$ be a randomized mutual exclusion algorithm,
whose expected makespan is $c$, where $c < n/2$.
We show that there exists an execution violating mutual exclusion.

%We use similar notation as in the previous proofs of Theorems~\ref{t:lower-not-feasible} 
%and~\ref{t:lower-log}, though the defined objects slightly differ from the previous ones.
%
Let $\cE^*_p$, for process $p$, be the set of all possible executions of the
first entry section of algorithm $\cR$ by process $p$ under the assumption that
it starts its first entry section in the global round $1$ and there is no other
process starting within the first $n/2$ rounds.  Note that during each
execution in $\cE^*_p$ process $p$ hears only noise (i.e., silence or
collision, which are indistinguishable due to the lack of collision detection)
from the channel when listening.  Observe also that the optimum algorithm needs
only one round to let process $p$ enter the critical section under the
considered adversarial scenario.  Therefore, by the probabilistic method, there
is an execution $\cE_p$ in set $\cE^*_p$ where process $p$ enters the critical
section within the first $n/2-1$ rounds. 
Let $\pi_{p}$ be the transmission
schedule of process $p$ during $\cE_p$.

Consider all sequences $\pi_p$ over all processes $0\le p < n$.
We construct execution $\cE$ contradicting mutual exclusion as follows.
First, we need to select a set of processes that start their first entry sections
in round $1$, while the others stay in the remainder section till at least round $n/2$.
Let $P_0=\{0,\ldots,n-1\}$.
For every non-negative integer $j$,
we define recursively 
\begin{eqnarray*}
P_{2j+1}
&=&
P_{2j} \setminus \left\{p\in P_{2j}: \exists_{i\in [1,n/2-1]} \left(\pi_p(i) = 1 
\ \& \ \forall_{q\in P_{2j},q\ne p} \ \pi_q(i)=0\right) \right\} \enspace, \\
P_{2j+2}
&=&
P_{2j+1} \setminus \left\{p\in P_{2j+1}: \exists_{i\in [1,n/2-1]} \left(|\pi_p| = i \ \& 
\ \forall_{q\in P_{2j+1}} \ |\pi_q|>i \right) \right\}
\enspace .
\end{eqnarray*}
Intuitively, set $P_{2j+1}$ is obtained from $P_{2j}$ by removing processes $p$
that could be single transmitters in some round in the interval $[1,n/2-1]$
while transmitting according to their schedules $\pi_p$.  Set $P_{2j+2}$ is
constructed by removing a process with the shortest transmission schedule, if
there is only one such process.  Observe that sequence $\{P_j\}_{j\ge 0}$ is
bounded and monotonically non-increasing (in the sense of set inclusion),
therefore it stabilizes on some set $P^*$.  Observe that
\begin{enumerate}
\item
$|P^*|\ge 2$, since for each round $i\in [1,n/2-1]$ there is at most one
process removed from some set $P_{2j}$ while constructing the consecutive set
$P_{2j+1}$ (after such removal no remaining process has $1$ in position $i$ of
its schedule) and at most one process removed from some set $P_{2j'+1}$ while
constructing the consecutive set $P_{2j'+2}$ (after such removal no remaining
process $p$ finishes its transmission schedule $\pi_p$ in round~$i$); as there
are $n/2-1$ considered rounds, at most $n-2$ processes can be removed throughout the construction;
\item
there is no round $i\in [1,n/2-1]$ such that there is only one process $p\in P^*$
satisfying $\pi_p(i)=1$; this follows from the fact that $P^*$ is a fixed point
of the sequence $\{P_j\}_{j\ge 0}$,
i.e., it does not change while applying the odd-step rule of the construction;
\item
there are at least two processes $p,q\in P^*$ with the shortest transmission
schedules $\pi_p,\pi_q$, i.e., $|\pi_p|=|\pi_q|$ and for every process $r\in
P^*$, $|\pi_r|\ge |\pi_p|$; this again follows from the fact that $P^*$ is a fixed point of the sequence $\{P_j\}_{j\ge 0}$,
i.e., it does not change while applying the even-step rule of the construction.
%otherwise the unique process of the shortest
%schedule would be removed while constructing some set $P_{2j+1}$. \darek{chyba $P_{2j+2}$?}
\end{enumerate}
Having subset $P^*$ of processes, the adversary starts first entry sections for
all processes in $P^*$ in the very first round, while she delays others (they
remain in the remainder section) by round $n/2$.  
Note that before round $1$ of the constructed execution $\cE$, 
a process $p\in P^*$ cannot distinguish $\cE$ from $\cE_p$, therefore
it may decide to do the same as in $\cE_p$, i.e., to set its first position of
transmission schedule to $\pi_p(1)$.  If this happens for all processes in
$P^*$, by the second property of this set there is no single transmitter in
round $1$, and therefore all listening processes hear the noise (recall that
silence is not distinguishable from collision in the considered setting).
This construction and the output of the first round can be inductively extended
up to round $|\pi_p|$, where $p\in P^*$ is a process with the shortest schedule
$\pi_p$ among processes in $P^*$.  This is because from the point of view of a
process $q\in P^*$ the previously constructed prefix of $\cE$ is not
distinguishable from the corresponding prefix of execution $\cE_q$; indeed, the
transmission schedules are the same and the feedback from the channel is
silence whenever the process listens.  Finally, by the very same reason, at the
end of round $|\pi_p|$ all processes $q\in P^*$ with $|\pi_q|=|\pi_p|$ are
allowed to do in $\cE$ the same action as in $\cE_q$, that is, to enter the
critical section.  By the third property of set $P^*$, there is at least one
such process $q\in P^*$ different than $p$.  This violates 
the exclusion property that should hold for the constructed execution $\cE$.
\end{proof}

%% file: random.tex
\section{Algorithms for the $\ep$-Mutual-Exclusion Problem}
\label{sec:random}

In this section, we present randomized algorithms solving the
$\ep$-mutual-exclusion problem for various scenarios, differing in the channel
capabilities (e.g., CD/no-CD, KN/no-KN, GC/no-GC).  The algorithms presented in
this section, work solely in entry sections, i.e., their exit sections are
empty; these algorithms guarantee only no-deadlock property.
However, in Section~\ref{sec:lockout}, we
show how to add exit section subroutines to most of our algorithms to
guarantee the no-lockout property while keeping bounded makespan.  In
our algorithms, we extend some techniques developed in the context of other
related problems, such as the wake-up problem~\cite{JS} and the leader
election problem~\cite{Wil}.

Throughout this section, we use the following notation.
We say that there is a {\em successful transmission} in a given round if in this round
one process transmits and other processes do not transmit.
By saying that a process {\em resigns}, we mean that it will not 
try to enter the critical section and will not attempt to transmit anything until 
another process starts the exit section.

%%%%%%%%%%%%%%%%%%%%%%%%%%%%%%%%%%%%%%%%%%%%%%%%%%%%

\subsection{Only Global Clock Available}
\label{sec:random:clock}
In the model with global clock, we modify the {\em Increase\_From\_Square} 
algorithm~\cite{JS}, which solves the wake-up problem.  The
purpose of our modification is to assure the stopping property. This is 
a nontrivial task in a scenario without collision detection and this property was not 
present in the original wake-up algorithm.  Intuitively, after one process
successfully transmits, it should enter the critical section.  However, first
of all it might not be aware that it succeeded. Second, between a successful
transmission and entering the critical section, some other processes may start
their entry sections. The details will be presented in the full version of this paper.

\begin{theorem}
\label{thm:global-clock}
There is an $\ep$-mutual-exclusion algorithm, using a modified algorithm 
{\em Increase\_From\_Square} as a subroutine for the entry section, 
with makespan $O(\log n \cdot \log (1/\ep))$ in the model without global clock.
\end{theorem}

%%%%%%%%%%%%%%%%%%%%%%%%%%%%%%%%%%%%%%%%%%%%%%%%%%%%%%

\subsection{Only Number of Processes Known}
\label{sec:random:kn}
In this scenario, we build our solution based on the \emph{Probability\_Increase} 
algorithm of~\cite{JS}.  In this algorithm, each process works in
$\Theta(\log n)$ phases, each lasting $\Theta(\log (1/\ep))$ rounds. In each
round of phase $i$, a process transmits with probability $2^{-i}$. 

\begin{lemma}[\cite{JS}]
\label{lem:JS-known-n}
If all processes use the algorithm {\em Probability\_Increase}
after being awaken, then there is a~successful transmission 
in time $k = O(\log n \cdot \log (1/\ep))$ with probability at least $1-\ep$.
\end{lemma}

We describe how to modify the \emph{Probability\_Increase} algorithm to meet
the requirements of $\ep$-exclusion. 
When a process enters the entry section, it first  switches to
the listening mode and stays 
in this mode for $k= O(\log n \cdot \log (1/\ep))$ rounds. 
If within this time it hears another process, it resigns.
Afterwards, the process starts to execute the {\em Probability\_Increase} algorithm. 
Whenever it is not transmitting, it listens, and when it hears a message from another
process, it resigns. 
After executing $k$ rounds of the listening mode and the following
$k$ rounds of {\em Probability\_Increase} without resigning, 
the process enters the critical section.
Using this algorithm, the following result can be proved.

\begin{theorem}
\label{thm:mod-prob-increase}
There is an $\ep$-mutual-exclusion algorithm, using a modified algorithm 
{\em Probability\_Increase} as a subroutine for the entry section, 
with makespan $O(\log n \cdot \log (1/\ep))$ in the KN model.
%The modified version of the {\em Probability Increase} algorithm solves the 
%$\ep$-mutual-exclusion problem in time $O(\log n \cdot \log (1/\ep))$,
%in the scenario where only the upper bound $n$ on the number of processes is known, 
\end{theorem}
\begin{proof}
%[Proof of Theorem~\ref{thm:mod-prob-increase}]
%Let $k$ be as defined above in the algorithm definition.  
%Consider the first round $t$ in which 
%%a 
%process enters its entry section.  
%
Let $k$ be as defined above in the algorithm definition.  
Let $t$ be a round in the execution in which there is at least one process in the entry section, no process in the exit or critical section,
and such that there was no process in the entry section in the previous round $t-1$.
Let $P$ be the set of processes which are in their entry sections at round $t+k$.  
First, we observe that processes which enter their entry section in round $t+k+1$ or later, 
i.e., all processes that are not in set $P$,
do not transmit  in the time period $[t,t+2k]$. By Lemma~\ref{lem:JS-known-n}, with
probability $1-\ep$, there is a process in $P$ which successfully transmits
at some round in $[t+k,t+2k)$.
Let $t+k\le r < t+2k$ be the first such round, and $p\in P$ be the process transmitting
successfully in round $r$.
Note that all other processes being in the entry section 
resign at this round, and all processes that start their entry sections after round $r$
do not transmit by round $r+k$.
Therefore, $p$ does not hear anything before it finishes its {\em Probability\_Increase} 
subroutine (in the next at most $k-1$ rounds after $r$), which implies that it enters the 
critical section by round $r+k-1<t+2k$.
\end{proof}

%%%%%%%%%%%%%%%%%%%%%%%%%%%%%%%%%%%%%%%%%%%%%%%%%%%%%%%%%%%%%%%

%%%%%%%%%%%%%%%%%%%%%%%%%%%%%%%%%%%%%%%%%%%%%%%%%%%%%%%%%%%%%%%%%%%%%%%%%%%%%%
\subsection{Only Collision Detection Available}
\label{sec:random:collision}

%\todo{check here}
%In the setting with collision detection, the Willard's algorithm~\cite{Wil} that we base our algorithm on
%is originally defined for the static scenario, i.e., it works only for the setting in which all processes which
%want to enter critical section, start their entry sections simultaneously. But a more important issue
%is that the Willard's algorithm assumes that the processes are able to transmit and listen at the same time
%and get an immediate feedback from the channel, i.e., they know whether in the current round there was 
%a silence, a transmission or a collision.

In this scenario, the main idea behind our algorithm is as follows. 
First, we show how to solve a {\em static case} of the $\ep$-mutual-exclusion problem, 
i.e., the case where there is a subset~$S$ of processes which start their entry sections at round $1$
and no process is activated later. Later, we show that we are then able to solve 
$\ep$-mutual-exclusion problem in (asymptotically) the same time.
In what follows, we assume that whenever a process does not transmit, it listens. 

To solve the static case, we first run a 
simple {\em Check\_If\_Single} subroutine, which, with probability 
at least $1-\ep$, determines whether there is one active processes or more. In the former case, 
this process may simply enter the critical section.
In the latter, we simulate Willard's algorithm~\cite{Wil}, which works in expected time 
$\log \log n + o(\log \log n)$. The simulation is required, as the original algorithm of~\cite{Wil}
assumes that each process can simultaneously transmit and listen in each round.
The idea of this simulation is that for each message sent, all listening processes 
acknowledge it in the next round. 

\begin{lemma}
\label{lem:ethernet-simulation}
If there are at least two active processes, it is possible to simulate one
round taken in the model in which a process may simultaneously transmit and
listen, in two rounds of our model in the setting with collision detection.
\end{lemma}

As mentioned above, another building block is a procedure {\em Check\_If\_Single}.  The algorithm assumes that there is a set of processes which start
this procedure simultaneously.  The procedure consists of $2 \cdot \log
(1/\ep)$ rounds.  In each odd round, process $i$ tosses a symmetric coin,
i.e., with probability $1/2$ of success,
to choose whether it transmits in the current round and listens in the next round,
or vice versa.
%and transmits in this round with probability $1/2$ and in the next (even) round otherwise. 
If the process never hears anything, it enters the critical section at the end of the procedure.

\begin{lemma}
\label{lem:check-if-single}
Assume $k$ processes execute the procedure {\em Check\_If\_Single}. 
If $k = 1$, then the only process enters the critical section.
If $k \geq 2$, then with probability $1-\ep$, no process enters the critical section.
\end{lemma}

\begin{proof}
%[Proof of Lemma~\ref{lem:check-if-single}]
The first claim holds trivially. For showing the second one, we fix 
an odd-even pair of rounds. Let $E$ denote the event
that there is a process, which does not hear anything in this pair of rounds.
For this to happen all processes running {\em Check\_If\_Single} 
have to transmit in the odd round or all have
to transmit in the even round. Thus, $\Pr[E] = 2 \cdot 1/2^k =1/2^{k-1} \leq 1/2$.
%$\Pr[\mathcal{E}] = 2 \cdot 1/2^k =1/2^{k-1} \leq 1/2$.  
Since the transmissions in different pairs of rounds are
independent, the probability that there exists a process which does not hear
anything during the whole algorithm, and thus enters the critical section, 
is at most $(1/2)^{\log (1/\ep)} = \ep$. 
\end{proof}

We may now describe an algorithm solving the static $\ep$-mutual-exclusion
problem.  Let $S$ be a subset of processes which simultaneously start their
entry sections.  In the first $2\log (1/\ep)$ rounds, the processes execute the
procedure {\em Check\_If\_Single}.  Then the processes that did not enter the
critical section, run a simulation of Willard's algorithm, 
as described in Lemma~\ref{lem:ethernet-simulation}.
The processes that transmit successfully, enter the critical section.
Using this algorithm, the following result can be proved.

\begin{theorem}
\label{thm:collision-detection}
In the scenario with collision detection, 
%the algorithm described above solves 
there is an algorithm solving
the static $\ep$-mutual-exclusion problem with expected makespan $O(\log \log n + \log (1/\ep))$.
\end{theorem}

\begin{proof}
%[Proof of Theorem~\ref{thm:collision-detection}]
Consider the algorithm described above, based on the
procedure {\em Check\_If\_Single}.
If there is only one process starting its entry section, it enters the critical
section right after the procedure {\em Check\_If\_Single} 
(which takes $O(\log (1/\ep))$ rounds).
If there is more than
one process, with probability $1-\ep$ they do not enter the critical section
after this procedure and they all simultaneously start the simulation of
Willard's algorithm. 
%(which is possible by Lemma~\ref{lem:ethernet-simulation}).
By the property of Willard's algorithm~\cite{Wil} and by Lemma~\ref{lem:ethernet-simulation}, 
in expectation there is a successful transmission in $O(\log \log n)$ rounds. 
\end{proof}

%%%%%%%%%%%%%%%%%%%%%%%%%%%%%%%%%%%%%%%%%%%%%%%%%%%%%%%%%%%%%%%

It remains to show that we may use an algorithm for static version of
$\ep$-mutual-exclusion to solve 
the general version of the $\ep$-mutual-exclusion problem. 
%without this restriction. 
The idea is to synchronize processes at the beginning, and then
to transmit a ``busy'' signal in every second round. New processes starting
their entry section note this signal and will not compete for the critical
section, until an exit section releases the shared channel.

\begin{theorem}
\label{thm:t-static-reduction}
If there exists an algorithm $\ALG$ for the static $\ep$-mutual-exclusion problem
with (expected) makespan $T$ in the model with collision detection, then there exists an
algorithm $\ALG'$ for the $\ep$-mutual-exclusion problem with (expected) makespan
$2 + 2 \cdot T$ in the same setting. 
\end{theorem}

%% file: no-lockout.tex
\section{Fairness}
\label{sec:lockout}

The algorithms shown in~\cite{CGKP} and Section~\ref{sec:random}
do not consider the no-lockout property,
i.e., it may happen that a process never gets out of its entry section, as
other processes exchange access to the critical section among themselves. 
We show how to modify algorithms satisfying no-deadlock property
(in particular, the algorithms from~\cite{CGKP}), so that the no-lockout
property is fulfilled.
Moreover, our transformation allows to express the (expected) makespan
of obtained fair protocols in terms of the (expected) makespan of the original weaker protocols.

Each process maintains an additional local counter of losses. When it starts its
entry section, it sets its counter to zero and whenever it loses the
competition for the critical section, 
i.e., when some other process enters the critical section,
it increases this counter by one. 
When a process enters its exit section, it becomes a guard: it helps
processes currently being in the entry section to choose one of them
with the highest loss counter. How high the loss counter can grow is
bounded by the number of processes in their entry sections at the moment when
the considered process entered its current entry section. Thus, also the time after
which the process will enter the critical section is bounded.
%We consider two models: when collision detection is available (CD) 
%or the number of processes is known (KN).

\begin{lemma}
\label{lem:lockout-cd}
If either collision detection is available or the number of nodes is known, 
it is possible to transform a mutual
exclusion algorithm with (expected) makespan $T$ into an algorithm, which also guarantees the no-lockout property and has (expected) makespan $O(T+\log n)$.

\end{lemma}

\begin{proof}%[Proof of Lemma~\ref{lem:lockout-cd}] 
%{\em (Sketch)}
Here we only describe a transformation for the CD scenario; 
the analysis and the variant for KN will appear in the full
version of the paper.
%In this proof we consider only the CD scenario; the variant for KN will appear in the full
%version of the paper.
%
Let $\ALG$ be a given subroutine for the entry section, satisfying no-deadlock
property. In order to use it for an entry procedure satisfying stronger no-lockout
property, we slow down algorithm $\ALG$ three times, by preceding each
original round by two additional rounds: in the first one the process
transmits signal $1$, while in the second one it only listens.
We call the obtained subroutine~$\ALG'$.

We also need the following selection subroutine. 
Assume there is a single guard and a subset (may be empty) of other processes,
called {\em competing processes}.
They all start the selection subroutine in the same round.
The goal is to elect one of the competing processes to enter the critical section.
%After the guard has made sure that there are some processes waiting to get into
%the critical section, it wants to choose one of them with the maximal loss counter. 
%The guard will synchronize communication in blocks of $3$ bits. The
%first two are always $10$ so, as a process starting its entry section can
%distinguish this situation from the previous ones, such processes do not
%participate in the current competition. 
The subroutine is partitioned into {\em blocks}, each consisting of three rounds.
In the first two rounds of each block only the guard transmits, and the signals are
$1$ and $0$, respectively. The purpose of these rounds is to assure that 
processes that start their entry sections later will not disturb the selection subroutine.
The competition, which is essentially a
binary search for the highest loss counter of the competing processes,
proceeds in the following phases. In the
$i$th block of the first phase all processes whose loss counter is at least
$2^i$ broadcast a $0$ (after the guard's $10$),
%\darek{moze lepiej broadcasts $100$, albo tak jak sugerowalem $000$?}
all other processes listen. The phase ends with a block
$i$ when silence is heard, thus all competitors and the guard know that the
highest loss counter is between $2^{i-1}$ and $2^i$. Then a binary search is
performed in additional $O(i)$ blocks in similar manner. Additional binary
search is performed to choose one process (the one with the minimum id) from
all processes with the same maximal number of losses.  

We now describe a procedure governing the exit section.
Recall that a process being in the critical section always broadcasts the critical
message to let others know that the channel is occupied. For the purpose of
this reduction and its analysis, we denote the~critical message by a single bit
$1$ (this is only technical assumption to simplify the proof arguments). 
When the process starts its exit section and becomes a guard, it
transmits a $0$ in the first round and listens in the second round. 
If the guard hears silence then it switches to the remainder section;
otherwise it participates in the selection subroutine described above.
%If a process starts its entry section,
%it can therefore recognize if any of the above-mentioned situations is happening:
%it hears either three $1$s in a row or a sequence $110$ in three consecutive rounds; 
%it can act accordingly.
%After the guard has entered its exit section, all other processes currently
%in their entry sections have to respond to this request by
%transmitting a $1$ in the second round, that is after hearing a $0$ after a
%series of $1$s. The guard thus gets to know if there is anyone who wants to
%enter the critical section. If not, it leaves the exit section.

Each process starting its entry section listens for three rounds.
If it hears silence during all these rounds,
it starts executing $\ALG'$ until some process enters the critical section
(it is guaranteed by no-deadlock property of $\ALG$, and can be extended
to $\ALG'$ as well); then it resets its state and starts again
its entry section procedure with round counter~$1$.
It~also resets its state and starts again with round counter~$1$
in case it hears 
%one silence and one no-silence
anything different from $1,1,1$ and $1,1,0$ 
during the first three rounds of listening.
In the remaining third case, i.e., when the process has heard
$1,1,1$ or $1,1,0$,
%the process may hear only a binary feedback in all three rounds. 
%It can be easily shown that this must be the sequence $1,1,0$ or $1,1,1$.
%Otherwise, if it hears ???, 
it keeps listening until the first round $t$, counting
from the first listening round in this run,
%beginning of the current entry section, 
such that the process has heard signals $1,1,0$ in rounds $t-2,t-1,t$, respectively.
It then transmits in round $t+1$ and starts the selection subroutine in round $t+2$.
\end{proof}

%\subsection{Known Number of Processes}

%\begin{lemma}
%\label{lem:lockout-kn}
%If the number of processes $n$ is known, it is possible to change a mutual
%exclusion algorithm, so that it also guarantees the no-lockout property at the
%cost of slowing down the original algorithm by a constant factor and with an
%additional cost of $O(\log n)$ rounds.
%\end{lemma}

By combining Lemma~\ref{lem:lockout-cd}
with the results from Section~\ref{sec:random}
and with the existing no-deadlock deterministic algorithms of~\cite{CGKP}, 
we obtain the following two conclusions.

\begin{corollary}
There exists a randomized algorithm with expected makespan $O(\log n + \log (1/\ep))$
solving the $\ep$-mutual-exclusion problem in the model in which collision
detection is available, and a randomized algorithm with makespan
$O(\log n \cdot \log (1/\ep))$ in the KN model.
\end{corollary}

\begin{corollary}
There exists a deterministic algorithm with makespan $O(\log n)$ solving
the mutual exclusion problem in the model in which collision detection is
available and a deterministic algorithm with makespan $O(n \log ^2 n)$ in the KN model.
\end{corollary}